\documentclass[12pt]{article}

\pdfoutput=1 
\usepackage{setspace}
\usepackage[T1]{fontenc}
\usepackage{lmodern}
\usepackage[colorlinks=true, pdfstartview=FitV, linkcolor=blue,citecolor=blue, urlcolor=blue]{hyperref}

\usepackage{microtype}
\usepackage[margin=1in]{geometry}
\usepackage[title]{appendix}
\usepackage[round]{natbib}
\usepackage[dvipsnames]{xcolor}
\usepackage{amsmath,amsthm,amssymb,ifthen,mathrsfs,amsfonts,bm,bbm}
\usepackage{multirow}
\usepackage{empheq}
\usepackage{mathtools}
\usepackage{mathabx}
\usepackage{subcaption}
\usepackage{graphicx}
\usepackage{accents}
\usepackage{enumitem}
\usepackage{bigints}
\usepackage{bookmark}
\usepackage{caption}
\usepackage[capposition=below]{floatrow}
\usepackage{verbatim}

\usepackage{pgf,tikz}
\usetikzlibrary{calc,math}
\usetikzlibrary{arrows,shapes,positioning,automata}
\usetikzlibrary{decorations.markings}
\usetikzlibrary{shapes.geometric}
\usetikzlibrary{patterns}

\newtheorem{theorem}{Theorem}
\newtheorem{lemma}{Lemma}
\newtheorem{proposition}{Proposition}
\theoremstyle{definition}
\newtheorem{corollary}{Corollary}
\newtheorem{definition}{Definition}

\newtheorem{remark}{Remark}

\def\BP{\mathsf{P}}
\def\BR{\mathbb{R}}

\def\d{\mathrm{d}}
\def\ee{\mathrm{e}}
\def\w{\omega}
\def\ep{\varepsilon}
\def\EE{\mathcal{E}}

\newcommand{\df}[1]{\textit{#1}}

\DeclarePairedDelimiter{\norm}{\lVert}{\rVert}

\makeatletter

\newcommand{\Rn}[1]{\expandafter\@slowromancap\romannumeral #1@}
\makeatother

\DeclareMathOperator*{\argmin}{arg\,min}

\numberwithin{equation}{section}

\title{Individual and Collective Welfare in Risk Sharing with Many States\thanks{We thank Itzhak Gilboa and Peter Klibanoff for questions that helped clarify our writing; and PJ Healy for some constructive skepticism. We are grateful to seminar audiences at Caltech, PET 2024, UC Berkeley, MSET 2024 (Durham University), D-TEA 2024, and the 2024 Latin American Econometric Society meetings, for questions and feedback.
Echenique (\href{mailto:fede@econ.berkeley.edu}{fede@econ.berkeley.edu}) is at the Department of Economics, UC Berkeley, and Pourbabaee (\href{mailto:farzad.pourbabaee@ucf.com}{farzad.pourbabaee@ucf.com}) is at the College of Business, University of Central Florida.}}

\author{
  Federico Echenique and Farzad Pourbabaee
}
\date{\today}
\renewcommand\footnotemark{}
\begin{document}

\maketitle

\begin{abstract}
We study efficient risk sharing among risk-averse agents in an economy with a large, finite number of states. Following a random shock to an initial agreement, agents may renegotiate. If they require a minimal utility improvement to accept a new deal, we show the probability of finding a mutually acceptable allocation vanishes exponentially as the state space grows. This holds regardless of agents' degree of risk aversion. In a two-agent multiple-priors model, we find that the potential for Pareto-improving trade requires that at least one agent’s set of priors has a vanishingly small measure. 
Our results hinge on the ``shape does not matter'' message of high-dimensional isoperimetric inequalities.  
\end{abstract}

\vspace{2cm}

\clearpage
\setstretch{1.35}
\tableofcontents
\clearpage

\interfootnotelinepenalty=10000
\medmuskip= 0.5mu plus 1.0mu minus 1.0mu
\section{Introduction}

An economy without aggregate risk is populated by a group of risk-averse agents.  The agents reach an efficient risk-sharing agreement that involves, for each of them, a state-contingent consumption plan.\footnote{Efficient risk sharing has a long history in economics, see e.g.\ \cite{borch1962} and~\cite{malinvaud1973}.} The agreement provides a certain level of welfare. Before executing the agreement, the agents learn of a random perturbation to the economy that induces aggregate, or systemic, uncertainty. Now the group may wish to renegotiate their initial agreement, but will only do so if they can guarantee themselves a minimal welfare improvement compared to what they had before the perturbation. There are transaction costs to renegotiating that must be covered. We ask: how likely is it that the agents can achieve this improvement? \textit{Our main result is that, when the number of states is large, the existence of a welfare-improving renegotiation is extremely unlikely. Its probability decreases exponentially in the number of states.}

The diagram on the left in  Figure~\ref{fig:introscitovsky} illustrates the situation, albeit in the case with two states ($d=2$) and two agents. Aggregate consumption is fixed at 1 (say 1 million dollars) in every state, represented by the vector $\w=(1,1)$. Aggregate consumption is state-independent, and thus implies no aggregate risk. The square defined by the origin and $\w=(1,1)$ is an Edgeworth box that describes all the allocations---all risk-sharing agreements---between the two agents. The Pareto optimal allocations are on the solid curve connecting the origin and $\w$, while $f=(f_1,f_2)$ is a particular Pareto optimal allocation (indicated by the tangency of the agents' indifference curves at the point $f$, with the usual Edgeworth-box convention that one agent's coordinate system has been rotated $180^{\circ}$ and its origin coincides with $\w$). For the allocation $f$, we may consider all the aggregate consumptions that can be decentralized among the two agents so as to obtain higher utility ``by $\ep$'' than in $f$.\footnote{An improvement by $\ep$ means a state-contingent consumption $g$ so that $(1-\ep)g$ is strictly preferred to $f_i$. If, for example, preferences are homothetic, this translates into a utility improvement of $\ep/(1-\ep)$\%.} This set of aggregate bundles consists of all the bundles on the dotted region to the north-east of the dashed curve passing through $\w$ (a curve that, at $\w$, has the same slope as the common tangent to agents' indifference curves at $f$).

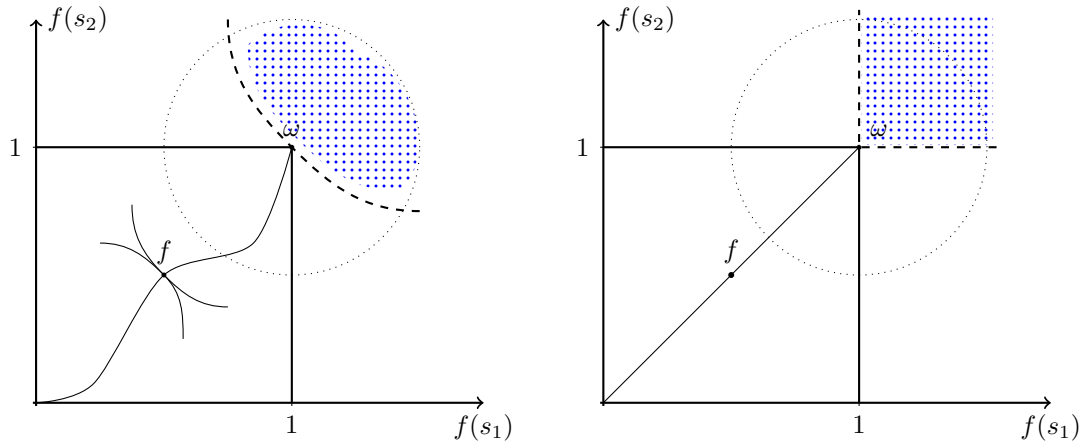
\begin{figure}[tth]
  \centering
\begin{tikzpicture}[scale=.85]
\draw[->, thick] (-.05,0) -- (7,0) node[below] {\footnotesize $f(s_1)$}; \draw[->, thick] (0,-.05) -- (0,6) node[right] {\footnotesize $f(s_2)$};

\coordinate (endow) at (4,4);

\draw[white,pattern=dots,pattern color=blue] (3.3,5.6) to
     [out=-90,in=135] (4.1,4.1) to
     [out=-45,in=180] (5.7,3.3) to
     [out=45,in=-15] (5,5.5) to 
[out=125,in=45] (3.3,5.6) ;

\draw[thick] (0,4) -- (endow) node[above] {\footnotesize $\omega$} -- (4,0);
\draw[fill] (endow) circle (.03);
\draw[thin] (4,0) -- (4,-0.05) node[below]  {\footnotesize $1$};
\draw[thin] (0,4) -- (-0.05,4) node[left]  {\footnotesize $1$};

\coordinate (split) at (2,2);
\draw[fill] (split) circle (.03) node[above] {\footnotesize $f$};

\draw[] (1.5,3.1) to [out=-90,in=135] (split);
\draw[] (3,1.5) to [out=180,in=-45] (split);
\draw[] (1,2.5) to [out=0,in=135] (split);
\draw[] (2.3,1) to [out=90,in=-45] (split);

\draw plot [smooth, thick]  coordinates {(0,0) (.9,0.3) (split) (3.4,2.5) (endow)};

\draw[dashed, thick] (3,6) to [out=-90,in=135] (endow);
\draw[dashed, thick] (6,3) to [out=180,in=-45] (endow);
\draw[dotted] (endow) circle (2);
\end{tikzpicture}
\quad
\begin{tikzpicture}[scale=.85]
\draw[->, thick] (-.05,0) -- (7,0) node[below] {\footnotesize $f(s_1)$}; \draw[->, thick] (0,-.05) -- (0,6) node[right] {\footnotesize $f(s_2)$};

\coordinate (endow) at (4,4);
\draw (0,0) -- (endow);
\draw[fill] (endow) circle (.03);
\draw[thin] (4,0) -- (4,-0.05) node[below]  {\footnotesize $1$};
\draw[thin] (0,4) -- (-0.05,4) node[left]  {\footnotesize $1$};

\draw[thick] (0,4) -- (endow) node[above right] {\footnotesize $\omega$} -- (4,0);

\coordinate (split) at (2,2);
\draw[fill, thick] (split) circle (.03) node[above] {\footnotesize $f$};


\draw[white,pattern=dots,pattern color=blue]  (6.1,4.03) -- (4.03,4.03) -- (4.03,6.05) -- (6.1,6.1) -- (6.1,4.03);
\draw[thick,dashed] (6.15,4) -- (endow) -- (4,6.15);
\draw[dotted] (endow) circle (2);

\end{tikzpicture}
    \caption{Risk sharing with two states and two agents.}
    \label{fig:introscitovsky}
\end{figure}

The aggregate consumption $(1,1)$ is subject to a random perturbation, an  aggregate shock $z$. Suppose that $z$ is symmetrically distributed around $(0,0)$; in fact, uniformly distributed among all $z$ with $\norm{z}\leq r$---for some $r>0$. Our results say that the probability that $z$ lies in the dotted region, and thus that aggregate risk can be decentralized among the agents in ways that improve individual welfare by $\ep$, is bounded above by $\ee^{-\ep^2 d/8r^2}$.  For example, when  $r=1$, $\ep=0.1$, and $d$ is the number of securities traded on NASDAQ, this bound is about 0.67\%. The point is that the probability of an agreeable renegotiation shrinks exponentially to zero as the number of states grows.

The result is surprising because it holds regardless of the agents' attitudes towards risk. For contrast, imagine instead that agents are extremely risk averse. Any contingent consumption is evaluated according to a worst-case scenario, which is determined by the smallest payoff in any given state. The picture on the right in Figure~\ref{fig:introscitovsky} illustrates the situation. Now an aggregate shock $z$ can only be distributed to make all agents better off if $z(s)>0$ in every state $s$, because if there is a state with less aggregate consumption, then some agent must consume less, and therefore be worse off, under the assumption of extreme risk aversion. The probability of such a perturbation is $1/2^d$, which also shrinks exponentially to zero. 

Comparing the pictures on the left and right of Figure~\ref{fig:introscitovsky} emphasizes the real message of our results. On the left, agents' limited risk aversion (the curvature of their indifference curves) leaves significant room for a welfare improvement after the shock $z$. When agents are close to risk-neutral, and $d=2$, the probability of a small $\ep$ improvement after the shock becomes arbitrarily close to $1/2$. As the number of states grows large, however, the probability inexorably decreases to zero. We can even assume that the allocation $f$ is supported by an equilibrium price vector $p$, and restrict attention to shocks $z$ with $p\cdot z>0$, the shocks that have some chance of being welfare-improving. Then the probability in the left Edgeworth box (with $d=2$) is arbitrarily close to $1$; while, again, as $d$ grows, it shrinks exponentially to zero.

\clearpage
Our paper presents three results on economies with many states.

The first result (Theorem~\ref{thm:walras_equil}) concerns a single agent's welfare. Imagine, again, a group of agents and an efficient allocation $f$. Suppose this time that the allocation is supported by some equilibrium prices; it is part of a Walrasian (or Arrow-Debreu) equilibrium. Given a shock $z$, the group looks for \textit{one} agent $i$ who is willing to absorb the shock, meaning that $f_i+z$ provides them with a welfare improvement of at least $\ep$ over $f_i$. The number of agents may be very large, and the group just needs one of them to take on $z$. The probability that there exists such an agent decreases, again, exponentially to zero with the number of states.

The second main result (Theorem~\ref{thm:scitovsky}) concerns an aggregate shock, and the ability to renegotiate an efficient risk-sharing agreement after the shock is realized. This is the result we have described using Figure~\ref{fig:introscitovsky}. In Proposition~\ref{prop:CRU}, we recast this result for allocations that represent inefficient risk sharing. The degree of inefficiency is measured by Debreu's coefficient of resource utilization.

Our third main result (Theorem~\ref{thm:small_belief_set}) is not explicitly about welfare, but instead about the attitudes towards uncertainty by agents who engage in efficient risk sharing. We
consider the same setting of economic exchange with no aggregate uncertainty, as in
our second result; but we strengthen the assumptions on preferences to focus on
utilities with multiple priors. Think, to fix ideas, of the max-min expected utility preferences of~\cite{gilboa1989maxmin}. Now we prove that, if the risk sharing agreement between the two agents can be improved in the Pareto sense, then at least one of the two must have a small set of prior beliefs. The stronger the level of Pareto improvement, the smaller the size of the set of prior beliefs. This means that at least one of the two agents must, in some sense, be close to being ambiguity neutral.

The sources of our, perhaps counterintuitive, results are some equally counterintuitive phenomena in high-dimensional probability. Our results rely on the study of concentration of measures, and specifically on the  implication of an \textit{isoperimetric inequality} for the separation of convex
sets. We explore the consequences of the ``shapes do not matter'' message from the literature in probability for economics. Key is the role of $\ep$, or of the coefficient of resource utilization, in our discussion above. These magnitudes provide a certain ``padding,'' or bound, on the degree of separation of the relevant convex sets. The phenomenon of concentration of measure
then has strong consequences for the volume of one of the sets being separated: one of the separated sets must have small volume. In our model, we can usually determine which of the sets must have small volume, and thus obtain the results we have described. Given its importance, we discuss the role and magnitude of $\ep$ in Section~\ref{sec:onepsilon}.

\section{The Model}\label{sec:themodel}

\subsection{Notations and Conventions}\label{sec:notationconventions}
Before presenting our model and main results, we lay down some of the basic definitions we shall make use of, as well as a few notational conventions.

If $S$ is a finite set, we denote by $\Delta S=\{\mu:S\mapsto \BR_+ \mid \sum_{s\in S}\mu(s)=1 \}$ the set of all probability measures on $S$.  We embed $\Delta S$ as a subset of $\BR_+^d$, and sometimes we refer to this probability simplex by $\Delta_d$. 

Let $A$ be a subset of a finite-dimensional normed vector space $(\BR^m, \norm{\cdot})$. The distance of a point $x\in \BR^m$ from $A$ is defined by $\mathsf{dist}(x, A) \coloneq \inf_{a \in A} \norm{x -a}$. By default, the norm used in the paper is the Euclidean $\ell_2$ norm. When a particular $p$-norm is used, we refer to the distance function by $\mathsf{dist}_p$, and the norm by $\norm{\cdot}_p$. For the $\ell_2$ norm, we usually omit the subscript~2. 

We represent the Euclidean open ball centered at $c$ and with radius $r$ by $\mathbb{B}(c,r)\coloneq \{x\in \BR^d:\norm{x-c}<r \}$. When the center is omitted, we take it to be the null vector $\bm{0}$. When the radius is omitted, we assume it is $r=1$. Thus $\mathbb B$ denotes the standard (open) unit ball. 

We denote the \textit{uniform} probability law on $\mathbb{B}(r)$ by $\BP^r$. We further define the class of (Borel) probability measures supported on $\mathbb{B}(r)$ that are absolutely continuous with respect to $\BP^r$, with their Radon-Nikodym derivative almost everywhere bounded above by a constant $\kappa$, as follows:
\begin{equation*}
    \mathcal{M}^r_\kappa \coloneq\left\{\BP^r_\kappa \in \Delta\left(\mathbb{B}(r)\right)\colon \frac{\d \BP^r_\kappa}{\d \BP^r}(z) \leq \kappa , \, \, \text{a.e.} \,\, z\in \mathbb{B}(r) \right\}\,.
\end{equation*}
For example, the $d$-dimensional standard Gaussian measure restricted to the $r$-ball has a bounded Radon-Nikodym derivative with respect to $\BP^r$, with the constant $\kappa \leq \ee^{r^2/2}$ (see Section~\ref{app:constrained_gaussian}).

For two subsets $A, B \subseteq \BR^m$, we define $\mathsf{dist}(A, B) \coloneq \inf\left\{ \norm{a - b} \colon a\in A, b\in B\right\}$. We denote the \emph{$\delta$-extension} of the subset $A$ by $A^\delta = \{x \colon \mathsf{dist}(x,A)<\delta\}$, which coincides with $A + \delta\, \mathbb{B}$.

Given a measurable subset $A\subseteq \BR^m$, we denote its Euclidean volume by $\mathsf{Vol}(A)$, which is equal to the Lebesgue integral of the indicator function of $A$ relative to the affine hull of $A$. For example, if $A$ is a $(m-1)$-dimensional surface in $\BR^m$, then $\mathsf{Vol}(A)$ refers to the surface area of $A$, as opposed to its $m$-dimensional volume (which is zero).

\subsection{Preferences and Uncertainty}
We consider a setting with uncertainty.\footnote{Using the reinterpretation of commodities in \cite{debreu1959theory}, our results have implications for many other economic environments. For example, they apply to the textbook version of the model, which assumes that consumption is in units of physically distinct goods.} Let $S$ be a finite set of possible \df{states of the world}, and $d$ be the number of states of the world;  $d\coloneq |S|$. 

Consequences are evaluated based on their monetary payoff in $\BR$. The set of \df{acts}, mappings from $S \to \BR$, is denoted by $\BR^d$, and individual acts are denoted by $f$ and $g$. We endow this set with its standard topology. 

Let $\succeq$ be a binary relation on $\BR^d$. As usual, we denote the strict part of $\succeq$ by $\succ$, and the associated indifference relation by $\sim$. We say that $\succeq$ is a (weakly monotone) \emph{preference relation} if it satisfies the following properties:
\begin{itemize}
    \item[--] (Weak Order): $\succeq$ is complete and transitive.
    \item[--] (Continuity): The upper and lower contour sets are closed.\\ That is, for every $f \in \BR^d$, the sets $\{g\colon g \succeq f\}$ and $\{g\colon f\succeq g\}$ are closed subsets of $\BR^d$.
    \item[--] (Monotonicity): For all $f, g \in \BR^d$ if $f(s) \geq g(s)$ for all $s \in S$, then $f \succeq g$. Furthermore, if $f(s)>g(s)$ for all $s \in S$, then $f \succ g$.
\end{itemize} 
The space of such preference relations on $\BR^d$ is denoted by $\mathcal{P}$. We say that $\succeq$ is a \textit{strongly monotonic} preference relation if, in addition to the weak order and continuity properties, it satisfies strong monotonicity: Namely, for all $f, g\in \BR^d$ if $f(s) \geq g(s)$ for all $s\in S$ and $f \neq g$, then $f \succ g$. We denote the space of strongly monotone preferences by $\mathcal{P}^{\text{sm}} \subset \mathcal{P}$.

Given a preference $\succeq$ and an act $f$, the set $\{g\colon g \succeq f\}$ is called the \df{upper contour set} of $\succeq$ at $f$. We say that a preference $\succeq$ is \df{convex} if its upper contour sets are convex, for all acts $f$. Convexity jointly with the weak order property and continuity imply that the \emph{strict} upper contour set, denoted by $\mathcal{U}_{\succeq}^{(0)}(f)\coloneq \{g\colon g \succ f\}$, is also convex. We refer to the space of convex preferences by $\mathcal{C}\subset \mathcal{P}$. We interpret convexity in our model as a basic property of risk, or uncertainty aversion~\citep{yaari1969some}.

Many well-known models in the theory of choice under uncertainty are special cases of convex preferences. Examples are risk-averse subjective expected utility (SEU), max-min expected utility \citep[MEU: see][]{gilboa1989maxmin}, multiplier preferences~\citep{hansen2001robust}, variational preferences~\citep{maccheroni2006ambiguity}, and smooth ambiguity aversion~\citep{klibanoff2005smooth}.

An approximate, or robust, notion of upper contour sets is key to our results.
\begin{definition}[$\ep$-upper contour set]
The $\ep$-approximate upper contour set of the preference  $\succeq$ at the act $f$ is defined by 
\begin{equation*}
    \mathcal{U}_{\succeq}^{(\ep)}(f) = \left\{g \in \BR^d\colon (1-\ep)g \succ f\right\}\,.
\end{equation*}
\end{definition}
When $\succeq$ is continuous and convex, the approximate upper contour set $\mathcal{U}_{\succeq}^{(\ep)}(f)$ is convex and open. Convexity is often of particular importance to us, as some of our arguments rely on a hyperplane separation theorem.

Approximate optimality is often expressed by means of the $\ep$-maximization of some numerical objective function---a \df{utility function} $u$ representing a preference $\succeq$. Observe that when $\succeq$ is homothetic (which is the case for many preferences used in applications), then $u$ may be taken to be homogeneous of degree one. As a consequence, our notion of approximate optimality translates directly into an approximation of utilities. Indeed, we may then write $\mathcal{U}_{\succeq}^{(\ep)}(f) = \left\{g \in \BR^d\colon (1-\ep)u(g) >   u(f)\right\}$. Then an act $g$ is in $\mathcal{U}_{\succeq}^{(\ep)}(f)$ if and only if it represents an $\ep/(1-\ep)$\% utility improvement over $f$.

\subsection{Exchange Economies} \label{sec:exchangeEconomies}
The set of agents is denoted by $I$ and a typical member of $I$ is referenced by index $i$. We assume throughout that $I$ is finite.

An \df{exchange economy} is a mapping $\EE\colon I \to \mathcal{P} \times \BR_+^d$, where $\mathcal{E}(i) = (\succeq_i, \omega_i)$. Each agent $i\in I$ is described by a preference relation  $\succeq_i$ on $\BR^d$, as well as an \df{endowment vector} $\w_i\in\BR_+^d$.  In an exchange economy, we use $\mathcal{U}_i^{(\ep)}$ to denote the upper contour set $\mathcal{U}_{\succeq_i}^{(\ep)}$.

Given an exchange economy $\EE$, the \df{aggregate endowment} is  $\omega \coloneq \sum_{i \in I}\omega_i$. A profile of acts across agents, say $f = (f_i)_{i\in I} \in \BR^{d \times I}$, is called an \df{allocation} if $\sum_{i\in I}f_i = \omega$. The space of all allocations is denoted by $\mathcal{F}_\omega$. We go back and forth between the textbook ``allocation'' terminology, and ``risk-sharing agreement,'' given our interpretation of consumption as a state-contingent monetary act.\footnote{In the textbook Walrasian setting, $\omega(s)$ represents the total available amount of good $s$ in the economy, and in an Arrow-Debreu model it represents the total contingent amount that agents can receive in state $s$. See Chapter~7 in~\cite{debreu1959theory}, or Chapter~19.C in~\cite{mas1995microeconomic}.}

\begin{definition}[$\ep$-Pareto optimality]
\label{def:approx_pareto_opt}
An allocation $f \in \mathcal{F}_\omega$ is called $\ep$-Pareto optimal if there is no allocation $g \in \mathcal{F}_\omega$ such that $g_i\in \mathcal{U}_i^{(\ep)}(f_i)$ for all $i \in I$. 
\end{definition} 
In words, an allocation $f$ is $\ep$-Pareto optimal if there is no redistribution of the aggregate endowment $\w=\sum_i f_i$ that would be strictly better for all agents, and that would remain strictly better for all agents after a fraction $\ep$ of consumption is ``shaved off'' in each state of the world. When $\ep = 0$, the definition coincides with the usual notion of \emph{weak Pareto optimality}. 
When we insist on the risk-sharing interpretation of our model, Pareto optimality means an ex-ante efficient risk-sharing agreement.

\begin{definition}[Walrasian equilibrium] An allocation $f=(f_i)_{i\in I}$ is called a Walrasian equilibrium allocation for the
exchange economy $\mathcal{E}$, if there exists a price vector $p \in \BR^d$ such
that $g_i \succ_i f_i$ implies that $p \,\cdot\, g_i > p \,\cdot\, \omega_i$,
and $p\,\cdot\, f_i = p\,\cdot\, \w_i$, for every $i \in I$.
\end{definition}

An exchange economy $\EE$ is \df{convex} if each preference relation $\succeq_i$ is convex, i.e., $\EE\colon I \to \mathcal{C}\times \BR_+^d$. Convexity is required for the basic theory of general equilibrium: existence of Walrasian equilibrium, as well as the second welfare theorem, relies on convex preferences.  Our Theorems~\ref{thm:scitovsky} and~\ref{thm:small_belief_set} will also rely on the convexity of agents' preferences (Theorem~\ref{thm:walras_equil} however, makes no such assumption, but it is a statement about Walrasian equilibrium allocations).

When $s\mapsto \sum_{i \in I}\omega_i(s)$ is constant, we say that $\EE$ exhibits \df{no aggregate uncertainty}. The aggregate endowment is then the same across all states of the world, i.e., $\omega = (\bar \omega, \ldots, \bar \omega)$. The assumption of no aggregate uncertainty is common in the study of risk sharing: see, for example,~\citet*[Ch.\ 19.C]{mas1995microeconomic},~\cite{allen1988optimal}  or, more closely related to our paper,~\cite*{Billot2000}. Our Theorem~\ref{thm:scitovsky} requires that there be no aggregate uncertainty.


\section{Main Results}\label{sec:main}
We proceed to state our main results. A broader discussion of our findings can be found in the introduction. We discuss more specific questions of interpretation in Section~\ref{sec:discussion}. The proofs of our results, including an overview of the methodology behind the proofs, can be found in Section~\ref{sec:proofs}.

\subsection{Walrasian Equilibrium}\label{sec:walras}

Our first result entails the following exercise: Fix a Walrasian equilibrium allocation, and draw a random perturbation from a ball of radius $r$ distributed according to some probability measure $\BP^r_\kappa \in \mathcal{M}^r_\kappa$.\footnote{Obviously the interesting measures $\BP^r_\kappa \in \mathcal{M}^r_\kappa$ involve $\kappa$ that is independent of $d$, which is the case for the uniform and conditional Gaussian measures. See our remarks in Section~\ref{sec:notationconventions}.} Consider the event that this random perturbation, when applied to \textit{any} agent's equilibrium consumption, results in a contingent consumption that is better by at least $\ep$ than their equilibrium consumption. The resulting event, that at least one agent is made better off, has a vanishingly small probability as $d$ increases.

\begin{theorem}\label{thm:walras_equil}
Let $\EE\colon I \to \mathcal{P} \times \BR_+^d$ be an exchange economy, $\omega = \left(\omega_i\right)_{i \in I}$ be the endowment profile, and $f = \left(f_i\right)_{i\in I}$ be a Walrasian equilibrium allocation. Assume that there exists $\tau >0$ such that $\omega_i \geq \tau \bm{1}$ and $f_i \geq \tau \bm{1}$ for all $i \in I$. Fix $r >0$ and let $z \sim \BP^r_\kappa \in \mathcal{M}^r_\kappa$. Then for every $\ep>0$,
\begin{equation}
\label{eq:walras_equil_bound}
\BP^r_\kappa\left((1-\ep)(f_i + z) \succ_i f_i \text{ for some } i \in I\right) \leq \kappa \,\ee^{-\ep^2 \tau^2 d/ 8r^2}\,.
\end{equation}
\end{theorem}

The proof of this theorem is in Section~\ref{sec:proofs}, as is the proof of the rest of the results we provide in this section. The proof relies on results from the literature on isoperimetric inequalities and the concentration of measure phenomenon, but the proof section attempts to be self-contained.\footnote{One simple intuition highlights the connection with the concentration of measures: Consider an economy with a single risk-neutral agent, and assume that all states are equally likely; so the agent's preferences are represented by $u(f)=\sum_s f(s)$. Then an $\ep$-improvement over a constant act $f(s)=1$ (which would be an equilibrium allocation in a Robinson Crusoe economy) requires that $(1-\ep)\sum_s z(s)>\ep d$, or that $\frac{1}{d}\sum_s z(s)> \ep/(1-\ep)$. In high dimensions, {\em despite the lack of independence} ($z$ is drawn from a sphere), the random variable $\frac{1}{d}\sum_s z(s)$ is tightly concentrated around $0$.}

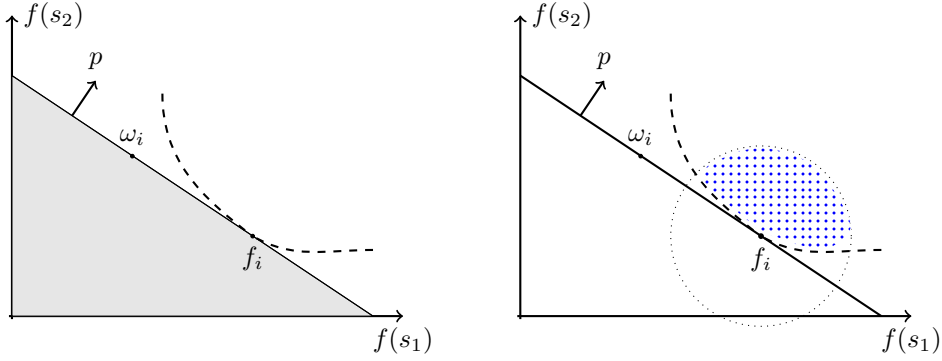
\begin{figure}[ttb]\centering
\begin{tikzpicture}[scale=.8]
  \draw[->, thick] (-.05,0) -- (6.5,0) node[below] {\footnotesize $f(s_1)$};
  \draw[->, thick] (0,-.05) -- (0,5) node[right] {\footnotesize $f(s_2)$};

  \draw[fill=gray!20] (0,0) -- (0,4) -- (6,0) -- cycle;
   \draw (0,4) -- (6,0);
   \coordinate (ef) at (4,1.33);
    \coordinate (endow) at (2,2.66);
    \draw[fill] (ef) circle (.03) node[below] {\footnotesize $f_i$};   
    \draw[fill] (endow) circle (.03) node[above] {\footnotesize $\w_i$};
    \draw[thick,->] (1,3.3333) -- ++ (90-34:.7) node[above] {\footnotesize $p$};

  \draw[dashed, thick] (2.5,3.7) to [out=-90,in=143] (ef);
  \draw[dashed, thick] (6,1.1) to [out=180,in=-30] (ef);

  \end{tikzpicture} \qquad
\begin{tikzpicture}[scale=.8]
  \draw[->, thick] (-.05,0) -- (6.5,0) node[below] {\footnotesize $f(s_1)$};
  \draw[->, thick] (0,-.05) -- (0,5) node[right] {\footnotesize $f(s_2)$};

   \draw [thick] (0,4) -- (6,0); 
  \coordinate (ef) at (4,1.33);
    \coordinate (endow) at (2,2.66);
    \draw[fill, thick] (ef) circle (.03) node[below] {\footnotesize $f_i$};   
    \draw[fill] (endow) circle (.03) node[above] {\footnotesize $\w_i$};
    \draw[thick,->] (1,3.3333) -- ++ (90-34:.7) node[above] {\footnotesize $p$};

  \draw[dashed, thick] (2.5,3.7) to [out=-90,in=143] (ef);
  \draw[dashed, thick] (6,1.1) to [out=180,in=-30] (ef);

\draw[dotted] (ef) circle (1.5);

\draw[white,pattern=dots,pattern color=blue] plot [smooth]
coordinates {(4.01,1.4) (5,1.1) (5.5,1.5) (5,2.4) (4.5,2.7) (4,2.8) (3.5,2.7) (3,2.3)
  (4.01,1.4)};

  \end{tikzpicture} 
    \caption{Individual welfare.}
    \label{fig:walrasian}
\end{figure}

To unpack the statement of Theorem~\ref{thm:walras_equil}, imagine an economy with two consumers, $I=\{1,2\}$. The situation of an individual agent $i$ is illustrated by Figure~\ref{fig:walrasian}, which shows the usual consumption-maximization diagram. In the figure, there are two states of the world ($d=2$), and an agent's budget set is depicted in gray. The budget is defined by an equilibrium price $p$, and by the agent's endowment $\w_i$. The picture shows the agent's equilibrium consumption $f_i$. The agent's indifference curve, tangent to the budget line at $f_i$, is depicted as a dashed curve. The upper contour set is the region to the northeast of the indifference curve that contains all consumptions that are preferred to $f_i$. If a random perturbation $z\in \mathbb{B}(r)$ is in $(\mathcal{U}_1^{(\ep)}(f_1) - f_1\})\cup (\mathcal{U}_2^{(\ep)}(f_2) - \{f_2\})$, then it means that either $(1-\ep)(z+f_1)$ provides a strict welfare improvement for agent 1 over $f_1$, or that $(1-\ep)(z+f_2)$ does this for agent 2 over $f_2$ (or that both things are true).

The consumer would be better off with more consumption in every state, but they would also be willing to accept a tradeoff of more consumption in one state than in another. Of course, the tradeoffs would have to be at terms of trade that are more favorable than equilibrium prices, and guided by the curvature of the agent's utility function (and its implied level of risk aversion). But the figure shows that there is a significant probability that a perturbation would leave the agent strictly better off, and the upper contour set contains much more than the bundles that are larger state-by-state. So does the dotted region. The theorem, actually, deals with the event that $z$ is in $(\mathcal{U}_1^{(\ep)}(f_1) - f_1\})\cup (\mathcal{U}_2^{(\ep)}(f_2) - \{f_2\})$, a potentially non-convex set that is usually strictly larger than the dotted region (and even larger when there are more than two agents).

The implications of the theorem may, therefore, be surprising for reasons that are familiar from our discussion in the introduction. The perturbations $z$ with $(1-\ep)(f_i + z) \succ_i f_i$ do not all feature $z>0$. In fact, the figure, as well as basic economic intuition, suggests that the curvature of the agents' utility functions should matter in calculating the probability in Theorem~\ref{thm:walras_equil}. So it is surprising that, for large $d$, the \emph{curvature, or risk attitude, does not matter} and we get an exponentially decreasing probability of a welfare improvement; the same as we would get if we were dealing with extremely risk-averse agents, and asking for a perturbation that is positive in all dimensions.

To put the bound in the theorem in context, consider an example with $\tau=r$, a uniform probability measure, and a 11\% welfare improvement ($\ep=0.1$). Then the probability of making at least one agent better off is at most $e^{-d/800}$. Financial economists often take the state space to equal the returns to a set of assets that are traded in the market: Hence, if we take $d$ to be the number of stocks, say, trading on the NASDAQ Exchange, then our bound from Theorem~\ref{thm:walras_equil} is about~0.67\%.\footnote{The number of stocks listed on NASDAQ is, according to \url{https://www.nasdaq.com/}, at least 4000.} 

\clearpage
\begin{remark}
We can express Theorem~\ref{thm:walras_equil} in terms of volumes. Specifically, when $z$ is drawn from the uniform measure $\BP^r$, the bound in~\eqref{eq:walras_equil_bound} is equivalent to
\begin{equation}
\label{eq:walras_equil_bound_vol}
    \frac{\mathsf{Vol}\left(\bigcup_{i \in I} \left(\mathcal{U}_i^{(\ep)}(f_i) - \{f_i\}\right) \, \cap \,  \mathbb{B}(r)\right)}{\mathsf{Vol}\left( \mathbb{B}(r)\right)} \leq \ee^{-\ep^2 \tau^2 d/ 8r^2}\,.
\end{equation}    
\end{remark}

Finally, we should emphasize that Theorem~\ref{thm:walras_equil} bounds the probability of an improvement to \emph{at least one} agent, while the bound is independent of the number of agents. We may then consider an economy with many agents and many states, and conclude that it is very unlikely that a given perturbation could make even one agent better off by at least $\ep$. 

Now, if we do focus on a particular agent, we immediately obtain:
\begin{corollary}\label{cor:walras_equil_bound}
Fix $r>0$ and let $z \sim \BP^r_\kappa \in \mathcal{M}^r_\kappa$. Under the hypotheses of Theorem~\ref{thm:walras_equil}, for every $\ep>0$, and for each $i \in I$,
\begin{equation}
\label{eq:walras_equil_bound2}
\BP^r_\kappa\left((1-\ep)(f_i+z) \succ_i f_i\right) \leq \kappa \, \ee^{-\ep^2 \tau^2 d/ 8r^2}\,.
\end{equation}
\end{corollary}    

Observe that the message of Theorem~\ref{thm:walras_equil} remains the same if we restrict attention to perturbations $z\in \mathbb{B}(f_i, r)$ for which $p\cdot z> p\cdot \w_i$. In a Walrasian equilibrium, the set $\bigcup_{i\in I}\left(\mathcal{U}_i^{(\ep)}(f_i) - f_i\right)$ lies in the half-space $\{z:p\cdot z\geq 0\}$, as $p\cdot f_i=p\cdot \w_i$, and a consumption that is affordable cannot provide an improvement in utility over $f_i$. So it seems natural to consider only $f_i+z$ that costs more than $i$'s income at the equilibrium prices. This subset of $\mathbb{B}(r)$ contains half its volume, and thus the message of Theorem~\ref{thm:walras_equil} and Corollary~\ref{cor:walras_equil_bound} remains unchanged. 

\subsection{Optimality and Resource Utilization}
\label{sec:pareto}
Our second result considers collective welfare improvements, and concerns outcomes that may not be Walrasian equilibria. An equilibrium allocation must be Pareto optimal, but the Pareto set is strictly larger. We ask about the possibility of a collective improvement after a shock in aggregate consumption, starting from an allocation that is a weak Pareto optimum.

In our next theorem, when we assume that there is no aggregate uncertainty, we normalize the aggregate endowment to be $\omega = (1,\ldots,1) = \bm{1}$.

\clearpage
\begin{theorem}\label{thm:scitovsky}
Let $\EE\colon I\to \mathcal{C}\times \BR_+^d$ be a convex exchange economy with no aggregate uncertainty, and an aggregate endowment of $\omega = \bm{1}$. Assume that $f$ is a weakly Pareto optimal allocation in $\EE$. For any $\ep>0$, let $\mathcal{V}^{(\ep)}(f) \coloneq  \sum_{i\in I} \mathcal{U}_i^{(\ep)}(f_i)$ be the Minkowski sum of the $\ep$-upper contour sets. Fix $r>0$ and let $z \sim \BP^r_\kappa \in \mathcal{M}^r_\kappa$. Then 
\begin{equation}
\label{eq:scitovsky}
\BP^r_\kappa\left(\omega  + z \in \mathcal{V}^{(\ep)}(f)\right) \leq \kappa \, \ee^{-\ep^2 d/ 8r^2}\,.
\end{equation}
\end{theorem}

The set $\mathcal{V}^{(0)}(f)$ of all aggregate bundles that may be disaggregated to make all agents strictly better off is called the \df{Scitovsky contour}; it is a basic notion in welfare economics (see \cite{de1942reconsideration} and \cite{samuelson1956social}), and plays a key role in the proof of the second welfare theorem. Here we focus on a version that requires $\ep$-welfare improvements.

\begin{remark}
Similar to Theorem~\ref{thm:walras_equil}, we can offer a version of the above result in terms of volume. Specifically, under the uniform measure $\BP^r$, the bound in~\eqref{eq:scitovsky} is equivalent to
\begin{equation}
\label{eq:scitovsky_vol}
   \frac{\mathsf{Vol}\left(\mathcal{V}^{(\ep)} \,\cap\, \mathbb{B}(\w, r)\right)}{\mathsf{Vol}(\mathbb{B}(\w, r))} \leq \ee^{-\ep^2 d/ 8r^2}\,.
\end{equation}
\end{remark}

It is also worth emphasizing that, just like in Theorem~\ref{thm:walras_equil}, the basic message of Theorem~\ref{thm:scitovsky} holds, even if we restrict attention to perturbations that result in a more expensive aggregate bundle. So, if $p$ is the supporting price vector obtained from the second welfare theorem, we can restrict to perturbations $z$ with $p\cdot z>0$. These are the perturbations that, according to neoclassical cost-benefit analysis (or the Kaldor criterion, see~\cite{graaff1967theoretical}), have some chance of being Pareto improving. The reason is the same as that previously discussed for Theorem~\ref{thm:walras_equil}.

A version of Theorem~\ref{thm:scitovsky} holds for allocations that are \emph{not} Pareto optimal. We make use of the measure of inefficiency introduced by~\cite{debreu1951coefficient}, the  \textit{coefficient of resource utilization}. To this end, consider an allocation $f$ in an exchange economy $\EE$ that is \emph{not} weakly Pareto optimal. This means that there is an alternative allocation of the aggregate endowment in $\EE$ that makes all agents strictly better off. In terms of the Scitovsky contour, this means that the aggregate endowment $\w=\sum_i\w_i$ lies in the set $\mathcal{V}^{(0)}$.

Debreu considers a minimal amount of aggregate resources (call it $\w^*$) that could be used to provide agents with the same utility as in $f$, and thinks of the gap between $\w$ and $\w^*$ as the inefficiency inherent in the allocation $f$. In Debreu's words, these are the ``non-utilized resources.'' He proposes to measure this gap by means of a ``distance with economic meaning:'' $p\cdot (\w-\w^*)$, where $p$ is an ``intrinsic price vector'' associated with $\w^*$. To obtain a measure that is, in a sense, scale independent, he actually works with the ratio of $p\cdot \w^*$ to $p\cdot \w$. Prices $p$ follow from an argument that is analogous to the second welfare theorem; in particular, they are not uniquely defined. 

Debreu's \df{coefficient of resource utilization} for an allocation $f=(f_i)_{i \in I}$ is defined as:
\[\mathrm{CRU}(f)
\coloneq \max_{\w^*\in \partial \overline{\mathcal V^{(0)}}} \frac{p(\w^*)\cdot \w^*}{p(\w^*)\cdot \w},
\] where $\partial \overline{\mathcal V^{(0)}}$ consists of the minimal elements of the
closure $\overline{\mathcal V^{(0)}}$ of $\mathcal V^{(0)}$ (meaning there is no smaller element in $\mathcal{V}^{(0)}$), and $p(\w^*)$ is a supporting
price vector at $\w^*$, what Debreu calls an intrinsic price vector. \cite{debreu1951coefficient} shows that $\mathrm{CRU}(f)$ is well defined (in particular, that it does not depend on the selection of prices $p(\w^*)$), that it is a number in $(0,1]$, and that $\mathrm{CRU}(f)<1$ when $f$ is Pareto dominated.

Then we obtain, following a similar argument as in Theorem~\ref{thm:scitovsky}: 

\begin{proposition}
\label{prop:CRU}
Let $\EE\colon I\to \mathcal{C}\times \BR_+^d$ be an exchange economy under the hypotheses of Theorem~\ref{thm:scitovsky}. Fix $r>0$ and let $z \sim \BP^r_\kappa$. If $f$ is not weakly Pareto optimal, and $\beta\coloneq\mathrm{CRU}(f)$ is its coefficient of resource utilization, then 
\begin{equation}
\label{eq:CRU}
\BP^r_\kappa\left(\omega + z \in \mathcal{V}^{(1-\beta^2)}(f)\right) \leq \ee^{-\left(\frac{1-\beta}{\beta}\right)^2 d/8r^2}\,.
\end{equation}   
\end{proposition}

Proposition~\ref{prop:CRU} quantifies the meaning of the coefficient of resource utilization. Debreu writes that one may think of $\mathrm{CRU}(f)$ as a percentage of national income, or GDP. But in an economy with a large state space, even a seemingly large amount of inefficiency---as measured by the coefficient of resource utilization---may not translate into a wide scope for welfare improvements by changing aggregate consumption. To use a similar numerical example from before, consider the case where $d$ is calculated from the number of stocks listed on NASDAQ. Suppose an inefficiency of 10\%, meaning a CRU of 0.9. Then the welfare improvement in Proposition~\ref{prop:CRU} is about 19\% and the bound on the probability of such an improvement is about 0.21\%. 

\subsection{Prior Beliefs and Welfare-Improving Trade}\label{sec:beliefs}

Our third result concerns agents with multiple priors, and the size of the sets of prior beliefs that they may possess. In previous sections, we quantified the upper contour sets of individual agents and their sums. Here we instead follow~\cite{yaari1969some} to interpret a vector that supports the upper contour sets, at some contingent consumption, as a prior belief. We depend on a relaxation of the equivalence between the existence of a common prior and Pareto efficiency of an allocation. This equivalence has been studied in a number of previous works \citep[e.g.,][]{Billot2000, ng2003duality, Rigotti_2008_ECMA,gilboasamuelsonschmeidler,GHIRARDATO2018730}. We shall follow~\cite{Rigotti_2008_ECMA} quite closely here.

We consider an exchange economy $\EE$ with no aggregate uncertainty. Importantly, from this point on we will not require convexity of preferences. The aggregate endowment is the same across all states of the world, i.e., $\omega = (\bar \omega, \ldots, \bar \omega)$. Also following the convention of~\cite{Rigotti_2008_ECMA}, in this section, we restrict the allocation space to nonnegative vectors. In this sense, we denote the unrestricted allocation space (defined in Section~\ref{sec:exchangeEconomies}) by $\mathcal{F}_{\bar \omega}$, and its nonnegative subset by $\mathcal{F}^+_{\bar \omega}\coloneq \mathcal{F}_{\bar \omega} \cap \BR_+^d$. We quantify its ``size'' by
\begin{equation}
\label{eq:rho}
    \rho\coloneq  2 \bar\omega^{-1} \max_{f \in \mathcal{F}^+_{\bar \omega}} \sum_{i \in I}\norm{f_i}\,.
\end{equation}

\cite{yaari1969some} defines a \emph{subjective belief} as a probability distribution that supports an agent's upper contour set at an act. When the upper contour has a kink (due, for example, to the ambiguity in preferences and its induced lack of differentiability), there will be multiple supporting vectors, and thus multiple beliefs (or multiple priors). In the spirit of~\cite{Rigotti_2008_ECMA}, we define the \emph{subjective belief set} as the set of all supporting vectors. 

Specifically, let $f$ be an act in $\BR_+^d$. The upper contour set of agent $i$ is $\{g\colon g \succeq_i f\}$, and the subjective belief set at $f$ is defined by
\begin{equation*}
    \mathcal{B}_i(f) = \left\{ \mu \in \Delta S\colon \mu \cdot g \geq \mu \cdot f \text{ for all } g \succeq_i f\right\}\,.
\end{equation*}
The set $\mathcal{B}_i(f)$ is a convex and closed (hence compact) subset of the $(d-1)$-dimensional probability simplex $\Delta S$ (or $\Delta_d$). One may also interpret a vector in $\mathcal{B}_i(f)$ as the vector of \emph{normalized prices} that supports the consumption of the act $f$.

We define the $\delta$-extension of the subjective belief set by 
\begin{equation*}
\mathcal{B}_i(f_i)^\delta \coloneq  \{\nu\in \Delta S\colon \inf_{\mu\in \mathcal{B}_i(f_i)} \norm{\nu-\mu}< \delta\}\,.
\end{equation*}

Our first result adapts the characterization of Pareto optimality in~\cite{Rigotti_2008_ECMA} to our setting, allowing for approximate Pareto domination.

\begin{proposition}\label{prop:approx_common_prior}
Let $\EE$ be an exchange economy with preferences $\succeq_i \in \mathcal{P}$ for all $i\in I$ and no aggregate uncertainty. Set $\delta = \ep/\rho$. If the allocation $f$ is $\ep$-Pareto dominated, then $\bigcap_{i \in I} \mathcal{B}_i(f_i)^\delta = \emptyset$.
\end{proposition}
In Proposition~\ref{prop:approx_common_prior}, we show that if an allocation is not \emph{approximately Pareto optimal} (measured by the parameter $\ep$), and thus there is room for welfare-improving trade, then the $\delta$-extension of the subjective belief sets share no common prior.
\begin{remark}
In Proposition~\ref{prop:approx_common_prior}, the extension of the subjective belief sets and the definition of $\rho$ in~\eqref{eq:rho} are both with respect to the $\ell_2$ norm. However, one can readily generalize by choosing an arbitrary $p$-norm for the belief sets, its conjugate $q$-norm for the definition of $\rho$, and the proof follows analogously. 
\end{remark}

Leveraging Proposition~\ref{prop:approx_common_prior}, in the following theorem we examine the volume of the prior sets as the number of states $d$ grows large. Before we state the result, we need to introduce some further notations. For a subset $J \subseteq I$, denote its complement by $J^c$, and define $\mathcal{B}_J(f_J) \coloneq \bigcap_{j \in J} \mathcal{B}_j(f_j)$. We often drop the allocation $f$ from the argument of subjective belief sets, when it is understood from the context. We further assume that  preferences are strongly monotonic (i.e., belonging to $\mathcal{P}^{\text{sm}}$). This assumption ensures that the dimension of each $\mathcal{B}_i$ is the same as $\Delta_d$.

\begin{theorem}\label{thm:small_belief_set}
Let $\EE\colon I\to \mathcal{P}^{\text{sm}}\times \BR_+^d$ be an exchange economy with no aggregate uncertainty. If the allocation $f \in \mathcal{F}^+_{\bar \omega}$ is $\ep$-Pareto dominated, then there exists a constant $c >0$, such that for every proper subset $J \subset I$, 
\begin{equation}
\label{eq:small_belief_set}
    \frac{\min\left(\mathsf{Vol}\left(\mathcal{B}_J\right), \mathsf{Vol}\left(\mathcal{B}_{J^c}\right)\right)}{\mathsf{Vol}\left(\Delta_d\right)} \leq \dfrac{1}{2}\,\ee^{-c \ep \sqrt{d}}\,.
\end{equation}
Moreover, the constant $c$ is universal: its value is independent of the primitives of the economy, the dimension $d$ and the parameter $\ep$.
\end{theorem}

The message of the theorem is clearest in a two-agent economy (for example the setting of~\cite{kajii2006agreeable} with MEU preferences). The existence of an $\ep$-Pareto improving trade is then possible \textit{only} when one of the two agents has a small set of subjective beliefs, as measured through its Euclidean volume. We expand on this interpretation in Section~\ref{sec:behavioralimplications}. 

More generally, in an economy with many agents, Theorem~\ref{thm:small_belief_set} states that, regardless of how we partition the agents into two groups, at least one group must have a small set of priors—meaning the volume of the intersection of their prior sets is small. If we think of the intersection as corresponding to a representative agent for the group (a notion that may be formalized), then this agent must be close to being ambiguity neutral: see again the discussion in Section~\ref{sec:behavioralimplications} for a sufficient condition under which the small volume conclusion leads to a behavioral property that bounds the degree of ambiguity aversion.

\section{Discussion}\label{sec:discussion}

\subsection{Behavioral Implications of Theorem~\ref{thm:small_belief_set}}\label{sec:behavioralimplications}

In this section we study the behavioral implications of Theorem~\ref{thm:small_belief_set}. For the purposes of this discussion, we consider a setting with two agents. The message of our theorem is that, if the current risk-sharing arrangement between the two agents is strongly Pareto dominated (namely if there exists a possibility of a strongly welfare improving ex-ante trade), then the volume of the underlying belief set of at least one agent must be very small. The \emph{smaller} this volume is, the \emph{closer} that agent is to being ambiguity neutral---a connection we seek to quantify in this section.\footnote{For an overview of ambiguity and ambiguity aversion, see~\cite{machina2014ambiguity} and~\cite{gilboa2016ambiguity}.}

It is well known that ambiguity aversion leads to less trade (see, for example,~\cite{dow1992uncertainty}). Theorem~\ref{thm:small_belief_set} offers a quantitative expression of this fact: that the possibility of an $\ep$-Pareto improving trade \textit{necessitates} small ambiguity aversion in high dimensions.

Our interpretation of the size of a set of priors is in line with the comparative notion of ambiguity aversion proposed by~\cite{Ghirardato_2004_JET}. Specifically, they show in the max-min expected utility setting of~\cite{gilboa1989maxmin}, that the ambiguity aversion of a decision maker decreases as her multiple prior set shrinks with respect to the set inclusion order. 

To establish the connection between ambiguity aversion and the volume of the decision maker's set of beliefs, or priors, we again appeal to the max-min expected utility setting of~\cite{gilboa1989maxmin}. In their model, there is a convex and compact set of priors $\Pi \subseteq \Delta S$, and the agent's choices over acts are guided by a utility function of the form \[
u(f) = \min_{\mu \in \Pi} f \,\cdot \, \mu,
\] for each act $f\in \BR_+^d$.  

We further assume that $\Pi$ has constant width, namely that the distance between any distinct parallel supporting hyperplanes of $\Pi$ (residing on the $(d-1)$-dimensional probability simplex) is constant. This assumption is crucial for our subsequent use of the result of~\cite{schramm1988volume}, but one can imagine a quantitative relaxation of the assumption of constant width and that it would translate into a modified bound on the degree of ambiguity aversion.

One may define the level of ambiguity aversion by the difference between the maximum and minimum expected utility of a \emph{normalized} act $f$ (i.e., $\norm{f}_2=1$) over $\Pi$, namely 
\begin{equation*}
    \theta(f)\coloneq \max_{\mu \in \Pi} f \cdot \mu - \min_{\mu \in \Pi} f \cdot \mu\,.
\end{equation*}
When $\Pi$ has constant width, $\theta(f)$ becomes constant in $f$, and we write $\theta(f) \equiv \theta$. Thus we can take $\theta$ as a measure of ambiguity aversion. Its meaning is the difference in how an ambiguity averse and ambiguity loving agent with the same set of priors would evaluate any given normalized act. A small degree of ambiguity aversion corresponds to a small difference between the ambiguity averse and ambiguity loving evaluations of an act.

In the next proposition we show how the upper bound on the relative volume in Theorem~\ref{thm:small_belief_set} means that $\theta$ vanishes as $d \to \infty$.
\begin{proposition}
\label{prop:theta_upper_bound}
Under the conditions of Theorem~\ref{thm:small_belief_set}, let $\Pi$ coincide with the set of priors with the smaller volume, and suppose that it has a constant width $\theta$. Then, there exists a universal constant $c>0$ such that
\begin{equation}
\label{eq:theta_upper_bound}
    \theta \leq 4 \,\ee^{-c \ep /\sqrt{d}} (d!)^{-1/2d}\,.
\end{equation}
\end{proposition}
\subsection{On the Interpretation and Magnitude of \texorpdfstring{$\ep$}.}\label{sec:onepsilon}

The presence of $\ep$ in a welfare improvement can be justified by appealing to robustness or transaction costs. When we talk about renegotiation, it is natural to imagine that reaching an agreement, with the requisite (unmodeled) communication requirements, is costly and is only worthwhile if the agents can obtain an improvement by some given margin. The number $\ep>0$ captures such transaction costs. Alternatively, we may restrict attention to some robust optimization criterion, which demands a minimum utility gain before moving away from a status quo. 

The interpretation of our results depends on the magnitude of two numbers: the number of states $d$, and the fraction $\ep$ accounted for in a utility improvement. When $d$ is large while $\ep$ remains constant, the probabilities discussed in Theorems~\ref{thm:walras_equil} and~\ref{thm:scitovsky}, and the volume in Theorem~\ref{thm:small_belief_set}, shrink to zero. So one could question our interpretation by arguing that we should use a small value of $\ep$ when the number of states $d$ is large. In particular, one might argue that we should impose that $\ep=O(1/\sqrt{d})$. We would disagree, however, essentially because we think of $\ep$ as a dimension-free fraction.

First, $\ep$ is expressed as a fraction of physical units of stage-contingent consumption. If we think that $f(s)$, for an act $f$,  is a monetary payment, then $\ep$ is a percentage of a monetary payment. It seems odd to impose a smaller percentage in monetary terms when the number of states is large than when it is small. For example, if we identify states with the number of assets in a market: is the meaning of a 5\% return different in a market with many assets than in a market with few assets? If we instead consider $\norm{\w}$ to be a measure of the ``size'' of the economy, then we may want to impose a value of $\ep$ that represents a constant fraction of $\norm{\w}$. For example, with the assumption that $\w=\mathbf{1}$ we have $\norm{\w}=\sqrt{d}$. Of course, the resulting $\ep$ would then grow with $d$, and only strengthen our results.

Second, taking $\ep=O(1/\sqrt{d})$ is problematic because it seems very hard to reconcile with the common practice of using a numerical objective function in calculating approximately optimal outcomes. As we discussed above, many applications make use of a homothetic preference and a resulting utility function that is homogeneous. Examples include the max-min representation in choice under uncertainty, or the Cobb-Douglas utility in consumer choice. In this case, the $\ep$ equals that tolerance level assumed in the agents' maximization problem: the number $\ep$ is then measured in ``utils,'' the same unit of account used for the utility function. Utils are, however, dimension-free: for example, if $u$ is a Cobb-Douglas utility, then $u(\mathbf{1})=1$ regardless of $d$. In a model with many states of the world, we would hardly be allowing for any relaxation in our notion of approximate optimality.

Finally, we should emphasize the reinterpretation of our results using the coefficient of resource utilization (see Proposition~\ref{prop:CRU}). The CRU is usually thought of as a fraction of national income, also a dimension-free measure, and one that might be expected to be constant, or even grow, with the size of the economy.


\section{Methodology and Proofs}\label{sec:proofs}

At a high level, the primary concept underpinning our results is the concentration of measure phenomenon. Every economist is familiar with the idea that, with a large sample, some sample statistics (such as the mean) are with high probability close to their population counterparts. In Euclidean spaces, dimension can play the role of sample size: in high-dimensional Euclidean spaces, and under relatively mild conditions, probability measures tend to concentrate as well. Specifically, for a probability measure $\mu$ on a bounded set $K\subseteq \BR^d$, and a subset $A$ that encompasses at least half of the probability space, namely $\mu(A) \geq 1/2$, the metric extension $A^\delta=\{z:\mathsf{dist}(z,A)<\delta \}$ covers a substantial and growing portion of the measure $\mu(K)$. In particular, the complement $1-\mu(A^\delta)$ diminishes rapidly, often exhibiting exponential decay with respect to the dimension $d$. These types of concentration bounds are commonly referred to as \emph{isoperimetric inequalities}. 

The importance of high-dimensional isoperimetric inequalities for our work lies in the \emph{independence} of the concentration rate from the subset $A$. Using the words we introduced earlier in the paper: shapes do not matter. The source of our results---that certain welfare improvements have a vanishingly small probability, independent of the shape of the agents' preferences---can be traced to basic observations in the theory of isoperimetric inequalities and concentration of measures. 

A consequence of the geometric concentration of measure is that, if two subsets $A$ and $B$ are separated with a positive distance, then as the dimension $d$ grows, the measure of at least one of them must be exponentially small. We apply this idea in the proofs of our results.

In this section, we begin by introducing the core inequality that relates to the  concentration of measures: the Brunn-Minkowski inequality. Then, we apply a variant of this inequality to prove Theorems~\ref{thm:walras_equil} and~\ref{thm:scitovsky}. Once we have shown these results, we present the preliminaries of isoperimetric inequalities, and employ them to prove Theorem~\ref{thm:small_belief_set}.
\subsection{Brunn-Minkowski Inequality}\label{sec:BMineq}
For two subsets $A, B \subseteq \BR^d$, their Minkowski sum is defined by $A+B \coloneq \left\{a + b\colon a \in A, b \in B\right\}$. The Brunn-Minkowski inequality provides a crucial connection between volumes and the Minkowski sum in Euclidean spaces. 

Let $A$ and $B$ be two non-empty compact subsets of $\BR^d$. The Brunn-Minkowski inequality~\citep{gardner2002brunn} claims that
\begin{equation}
\label{eq:BM_ineq}
\mathsf{Vol}(A+B)^{1/d} \geq \mathsf{Vol}(A)^{1/d} + \mathsf{Vol}(B)^{1/d}\,.
\end{equation}
If one makes the additional assumption that $A$ and $B$ are restricted to convex subsets, then the inequality binds if and only if $A$ and $B$ are homothetic (that is, one is the translated and scaled version of the other). This inequality implies the concavity of the volume operator with respect to the Minkowski sum. There is a \emph{dimension-free} version of this inequality that often proves more useful. In particular, for $\lambda \in [0,1]$, inequality~\eqref{eq:BM_ineq} implies that
\begin{equation*}
\mathsf{Vol}(\lambda A+ (1-\lambda) B)^{1/d} \geq \lambda \mathsf{Vol}(A)^{1/d} + (1-\lambda) \mathsf{Vol}(B)^{1/d}\,.
\end{equation*}
Applying the arithmetic geometric inequality to the above provides the following dimension-free version:
\begin{equation}
\label{eq:dim_free_BM_inequ}
\mathsf{Vol}(\lambda A+ (1-\lambda) B) \geq  \mathsf{Vol}(A)^{\lambda}\, \mathsf{Vol}(B)^{1-\lambda}\,.
\end{equation}
We henceforth refer to this inequality as the BM inequality. A useful application of the BM inequality is the following lemma, which provides an upper bound for the minimum volume of two positively distanced subsets (its proof can be found in~\cite{artstein2015asymptotic}, but we include it here for completeness).
\begin{lemma}\label{lem:high_dim_sep}
Assume $A$ and $B$ are Borel subsets of $\mathbb{B}(r)$, and $\mathsf{dist}(A,B) \geq \delta$. Then, 
\begin{equation}
\label{eq:min_vol}
    \frac{\min\{\mathsf{Vol}(A), \mathsf{Vol}(B)\}}{\mathsf{Vol}(\mathbb{B}(r))} \leq \ee^{-\delta^2 d/8r^2}\,.
\end{equation}
\end{lemma}
\begin{proof}
Since the volume of any Borel set can be approximated arbitrarily closely by the inner measure of its closed subsets, we can assume without loss of generality that $A$ and $B$ are closed, and hence compact. By the parallelogram law for the $\ell_2$-norm if $a\in A$ and $b\in B$ then
\begin{equation*}
    \norm{a+b}^2 = 2\norm{a}+ 2 \norm{b}^2 - \norm{a - b}^2 \leq 4r^2 - \delta^2\,, 
\end{equation*}
where the inequality holds because $a, b \in \mathbb{B}(r)$ and $\norm{a-b} \geq \delta$. Hence, it follows that 
\begin{equation*}
    \frac{A + B}{2} \subseteq \sqrt{1-\frac{\delta^2}{4r^2}} \, \mathbb{B}(r)\,,
\end{equation*}
and therefore,
\begin{equation*}
    \mathsf{Vol}\left(\frac{A + B}{2}\right) \leq \left(1-\frac{\delta^2}{4r^2}\right)^{d/2} \mathsf{Vol}(\mathbb{B}(r)) \leq \ee^{-\delta^2 d /8r^2}\,\mathsf{Vol}(\mathbb{B}(r))\,.
\end{equation*}
Setting $\lambda=1/2$ in~\eqref{eq:dim_free_BM_inequ} and using the above inequality justifies the claim in~\eqref{eq:min_vol}.
\end{proof}
Therefore, as the dimension grows, two subsets in the $\ell_2$ ball with a bounded radius will have positive distance from each other only if at least one of them has a very small volume. Of course, the greater the distance, the smaller the implied volume.
\subsection{Proof of the Results in Sections~\ref{sec:walras} and~\ref{sec:pareto}}

We proceed with the proof of our first two theorems. In both cases, using the optimality or the equilibrium property of the allocation, we apply a type of convex separation argument. In the first theorem, the separation is provided by the given equilibrium price. In the second theorem, the separation argument follows the ideas used in the standard proofs of the second welfare theorem. Consequently, as we argue below, the approximate versions of upper contour sets stay at a positive distance from their separated feasible counterparts. An application of Lemma~\ref{lem:high_dim_sep} to appropriately chosen subsets implies the probability bounds in Theorems~\ref{cor:walras_equil_bound} and~\ref{thm:scitovsky}.

We begin by laying down some terminology. For a vector $p \in \BR^d$, $p\neq \bm{0}$, and a constant $b$, we define two half-spaces:
\begin{equation*}
    \begin{gathered}
        H^+(p\,;\,b) \coloneq  \left\{x\in \BR^d\colon p \, \cdot \, x \geq b\right\}\,,\\
        H^-(p\,;\,b) \coloneq  \left\{x \in \BR^d\colon p \, \cdot \, x \leq b\right\}\,,
    \end{gathered}
\end{equation*}
that are, respectively, called upper and lower half-spaces. One can readily verify that the $\ell_2$ distance between the two parallel half-spaces $H^+(p\,;\,b_2)$ and $H^-(p\,;\,b_1)$, where $b_2 > b_1$, is equal to
\begin{equation}
\label{eq:paral_hp_dist}
    \mathsf{dist}\left(H^+(p\,;\,b_2), H^-(p\,;\,b_1)\right) = \frac{b_2 - b_1}{\norm{p}}\,.
\end{equation}

\subsubsection*{Proof of Theorem~\ref{thm:walras_equil}}
Since $f= (f_i)_{i\in I}$ is a Walrasian equilibrium allocation in $\EE$, and preferences are monotone, there exists a price vector $p \in \BR_+^d$ such that $p \, \cdot \,  g_i > p \, \cdot \, \omega_i$ for all $i\in I$ and $g_i \in \mathcal{U}_i^{(0)}(f_i)$. Next, observe that if $g \in \mathcal{U}_i^{(\ep)}(f_i)$ then $(1-\ep)g \in \mathcal{U}_i^{(0)}(f_i)$ and therefore $p \cdot \left((1-\ep)(g-\omega_i) - \ep \omega_i\right) > 0$. This in turn implies that
\begin{equation*}
    p\, \cdot \, (g-\omega_i) > \frac{\ep p\,\cdot \, \omega_i}{1-\ep} > \ep p \,\cdot\, \omega_i \geq \ep \tau \norm{p}_1\,,
\end{equation*}
where the last inequality holds because $p \in \BR_+^d$ and $\omega_i \geq \tau \bm{1}$. 

Therefore, since $i\in I$ and $g_i \in \mathcal{U}_i^{(\ep)}(f_i)$ were arbitrary, by the above inequality $\mathcal{U}_i^{(\ep)}(f_i) -\{\omega_i\} \subseteq H^+\left(p ; \ep \tau \norm{p}_1\right)$ for all $i \in I$. Let us define $\mathcal{Q} \coloneq  \bigcup_{i \in I} \left(\mathcal{U}_i^{(\ep)}(f_i) -\{\omega_i\}\right)$.\footnote{The set $\mathcal{Q}$ was originally used by~\cite{debreu1963limit} to prove core convergence, and by~\cite{barmanechen23} to characterize approximate Walrasian equilibria.} Then $\mathcal{Q}\subseteq H^+\left(p ; \ep \tau \norm{p}_1\right)$, so for an arbitrary $r>0$, one has
\begin{equation*}
    \begin{gathered}
        \mathsf{dist}\left(\mathcal{Q} \cap \mathbb{B}(r) ,  H^{-}(p;0) \cap \mathbb{B}(r)\right) \geq \mathsf{dist}\left(H^+\left(p ; \ep \tau \norm{p}_1\right) \cap \mathbb{B}(r) ,  H^{-}(p;0) \cap \mathbb{B}(r)\right)\\
        \geq \mathsf{dist}\left(H^+\left(p ; \ep \tau \norm{p}_1\right)  ,  H^{-}(p;0) \right) = \ep \tau \frac{\norm{p}_1}{\norm{p}}\geq \ep \tau\,.
    \end{gathered}
\end{equation*}
The equality above follows from~\eqref{eq:paral_hp_dist}, and the last inequality holds because $\min_{p \neq \bm{0}} \norm{p}_1 / \norm{p}_2 = 1$, namely the minimum of $\norm{p}_1 / \norm{p}_2$ is achieved on the standard unit basis vectors, and is equal to $1$. 

Now set $A\coloneq  \mathcal{Q} \cap \mathbb B(r) $ and $B\coloneq  H^{-}(p;0) \cap \mathbb{B}(r)$. By the above inequality $\mathsf{dist}(A,B) \geq  \ep \tau$. Since $p$ is a nonzero vector in $\BR_+^d$, the subset $B$ covers at least half of the volume of $\mathbb{B}(r)$. So it must be that $\mathsf{Vol}(A) \leq \mathsf{Vol}(B)$. Therefore, Lemma~\ref{lem:high_dim_sep} implies that $\mathsf{Vol}(A) / \mathsf{Vol}(\mathbb{B}(r)) \leq \ee^{-\ep^2\tau^2 d /8r^2}$, namely:
\begin{equation*}
\frac{\mathsf{Vol}\left(\mathcal{Q} \,\cap\, \mathbb{B}(r)\right)}{\mathsf{Vol}(\mathbb{B}(r))} \leq \ee^{-\ep^2 \tau^2 d/ 8r^2}\,,
\end{equation*}
and thereby
\begin{equation*}
\frac{\mathsf{Vol}\left(\bigcup_{i \in I} \left(\mathcal{U}_i^{(\ep)}(f_i) - \{\omega_i\}\right) \, \cap \,  \mathbb{B}(r)\right)}{\mathsf{Vol}\left( \mathbb{B}(r)\right)} \leq \ee^{-\ep^2 \tau^2 d/ 8r^2}\,.
\end{equation*}
Now observe that if $f=(f_i)_{i\in I}$ is a Walrasian equilibrium for the exchange economy $\mathcal{E}$, it is also a Walrasian equilibrium for the exchange economy $\mathcal{E'}$ that is identical to  $\mathcal{E}$ except that each agent $i$'s endowment is $\omega'_i=f_i$. Therefore, we can replace $\omega_i$ in the above inequality with $f_i$, and obtain the volume bound in~\eqref{eq:walras_equil_bound_vol} that is equivalent to 
\begin{equation*}
\BP^r\left((1-\ep)(f_i + z) \succ_i f_i \text{ for some } i \in I\right) \leq \ee^{-\ep^2 \tau^2 d/ 8r^2}\,.
\end{equation*}
Since $\frac{\d \BP^r_\kappa}{ \d \BP^r} (z) \leq \kappa$ almost everywhere on $z \in \mathbb{B}(r)$, the claim in~\eqref{eq:walras_equil_bound} follows from the above bound.\qed

\subsubsection*{Proof of Corollary~\ref{cor:walras_equil_bound}}
The union bound and~\eqref{eq:walras_equil_bound_vol} imply that
\begin{equation*}
\frac{\mathsf{Vol}\left( \left(\mathcal{U}_i^{(\ep)}(f_i) - \{f_i\}\right) \, \cap \,  \mathbb{B}(r)\right)}{\mathsf{Vol}\left( \mathbb{B}(r)\right)} \leq \ee^{-\ep^2 \tau^2 d/ 8r^2}\,, \forall i \in I\,.
\end{equation*}
Since the Euclidean volume is translation invariant, we can shift the subsets in the above inequality by $f_i$, and thus obtain:
\begin{equation}
    \frac{\mathsf{Vol}\left(\mathcal{U}_i^{(\ep)}(f_i) \,\cap\, \mathbb{B}(f_i, r)\right)}{\mathsf{Vol}(\mathbb{B}(f_i, r))} \leq \ee^{-\ep^2 \tau^2 d/ 8r^2}\,, \forall i \in I\,.
\end{equation}
Subsequently, a measure change from the uniform law $\BP^r$ to $\BP^r_\kappa$ implies~\eqref{eq:walras_equil_bound2}.\qed

\subsubsection*{Proof of Theorem~\ref{thm:scitovsky}} Since $f=(f_i)_{i \in I}$ is weakly Pareto optimal, then there is no allocation $g \in \mathcal{F}_{\bm{1}}$ such that $g_i \succ_i f_i$ for every $i \in I$. That is $\mathcal{F}_{\bm{1}} \cap \prod_{i \in I} \mathcal{U}_i^{(0)}(f_i) = \emptyset$. That in turn means the normalized endowment vector $\omega = \bm{1}$ is disjoint from $\mathcal{V}^{(\ep)}\coloneq \sum_{i \in I} \mathcal{U}_i^{(\ep)}(f_i)$ for all $\ep \geq 0$. Since preferences are convex the approximate upper contour sets are convex, so is their sum $\mathcal{V}^{(\ep)}$. Therefore, by the hyperplane separation theorem and monotonicity of preferences there exists a nonzero vector $p \in \BR^d_+$ such that $p \cdot v \geq p \cdot \bm{1} = \norm{p}_1$ for all $v \in \mathcal{V}^{(0)}$. That is $\bm{1} \in H^-(p; \norm{p}_1)$ and $\mathcal{V}^{(0)} \subseteq H^+(p; \norm{p}_1)$.

Now suppose $v \in \mathcal{V}^{(\ep)}$. Then, there are $g_i \in \mathcal{U}_i^{(\ep)}$ for every $i \in I$, such that $v=\sum_{i \in I} g_i$ and $(1-\ep)g_i \succ_i f_i$ . For each $i\in I$, it holds that $(1-\ep)g_i \in \mathcal{U}_i^{(0)}$ and hence $(1-\ep)\sum_{i\in I} g_i \in \mathcal{V}^{(0)}$. Consequently, one has $p \cdot (1-\ep)v \geq \norm{p}_1$. That in turn implies $v \in H^+\left(p; \norm{p}_1/(1-\ep)\right)$, and thereby $\mathcal{V}^{(\ep)} \subseteq H^+\left(p; \norm{p}_1/(1-\ep)\right)$. As a result of this set inclusion, for an arbitrary $r>0$, one obtains that
\begin{equation*}
\begin{gathered}
    \mathsf{dist}\left(\mathcal{V}^{(\ep)} \cap \mathbb{B}(\bm{1}, r) \, , \, H^-(p\,;\, \norm{p}_1) \cap \mathbb{B}(\bm{1}, r)\right)  \geq \\ \mathsf{dist}\left(H^+\left(p\,;\, \frac{\norm{p}_1}{1-\ep}\right) \cap \mathbb{B}(\bm{1}, r) \, , \, H^-\left(p\,;\, \norm{p}_1\right) \cap \mathbb{B}(\bm{1}, r)\right) \\
    = \mathsf{dist}\left(H^+\left(p\,;\, \frac{\norm{p}_1}{1-\ep}\right)  \, , \, H^-(p\,;\, \norm{p}_1) \right)\,,
\end{gathered}
\end{equation*}
where the equality holds because the two half-spaces are parallel. Their distance by~\eqref{eq:paral_hp_dist} is equal to $\ep \norm{p}_1 /(1-\ep)\norm{p}$. Therefore, we arrive at
\begin{equation*}
    \mathsf{dist}\left(\mathcal{V}^{(\ep)} \cap \mathbb{B}(\bm{1}, r) \, , \, H^-(p\,;\, \norm{p}_1) \cap \mathbb{B}(\bm{1}, r)\right)  \geq \frac{\ep \norm{p}_1}{(1-\ep)\norm{p}} \geq \ep\,,
\end{equation*}
where the last inequality follows as before, because $\min_{p \neq \bm{0}} \norm{p}_1 / \norm{p}_2 = 1$. Now set $A\coloneq  \mathcal{V}^{(\ep)} \cap \mathbb{B}(\bm{1},r)$ and $B \coloneq  H^-(p\,;\, \norm{p}_1) \cap \mathbb{B}(\bm{1},r)$. By the above inequality one has $\mathsf{dist}(A,B) \geq \ep$. Since $p$ is a nonzero vector in $\BR_+^d$, the subset $B$ covers at least half of the volume of $\mathbb{B}(\bm{1},r)$. So it must be that $\mathsf{Vol}(A) \leq \mathsf{Vol}(B)$. Therefore, Lemma~\ref{lem:high_dim_sep} implies that $\mathsf{Vol}(A) / \mathsf{Vol}(\mathbb{B}(\bm{1},r)) \leq \ee^{-\ep^2 d /8r^2}$, thus proving the volume bound in~\eqref{eq:scitovsky_vol} that is equivalent to~\eqref{eq:scitovsky} after a change of measure from $\BP^r$ to $\BP^r_\kappa$.\qed

\subsubsection*{Proof of Proposition~\ref{prop:CRU}}
Let $\beta\coloneq \mathrm{CRU}(f)$. Since $f$ is not weakly Pareto optimal, then $\beta<1$. By \cite{debreu1951coefficient}, it holds that $\beta \bm{1}\in \partial\overline{\mathcal{V}^{(0)}}$. Observe that for every $i$, one has 
\begin{equation*}
    \frac{1}{\beta} \,\mathcal{U}_i^{(0)}(f_i) = \mathcal{U}_i^{(1-\beta)}(f_i)\,.
\end{equation*}
Therefore $\beta^{-1}\, \mathcal{V}^{(0)}(f) = \mathcal{V}^{(1-\beta)}(f)$, and $\bm{1}$ becomes a minimal element of $\mathcal{V}^{(1-\beta)}(f)$. Following the same recipe as in the proof of Theorem~\ref{thm:scitovsky}, let $p$ be the supporting vector of $\mathcal{V}^{(1-\beta)}$ at $\bm{1}$, and consider the $r$-ball centered at $\beta \bm{1}$. Set $A \coloneq \mathcal{V}^{(1-\beta)} \cap \mathbb{B}(\beta \bm{1}, r) \subseteq H^+\left(p; \norm{p}_1\right)$ and $B \coloneq H^{-}\left(p; \beta \norm{p}_1\right) \cap \mathbb{B}(\beta \bm{1}, r) \subseteq H^{-}\left(p; \beta \norm{p}_1\right)$. Then, $\mathsf{dist}(A, B) \geq 1-\beta$, and $\mathsf{Vol}(B) \geq \mathsf{Vol}(A)$. Therefore by Lemma~\ref{lem:high_dim_sep}, $\mathsf{Vol}(A) / \mathsf{Vol}(\mathbb{B}(\beta \bm{1},r)) \leq \ee^{-(1-\beta)^2 d/ 8r^2}$, that is equivalent to
\begin{equation*}
    \BP^r\left(\beta \bm{1} + z \in \mathcal{V}^{(1-\beta)}(f)\right) \leq \ee^{-(1-\beta)^2 d/ 8r^2}\,.
\end{equation*}
Because $\beta^{-1}\mathcal{V}^{(1-\beta)}(f) = \mathcal{V}^{(1-\beta^2)}(f)$, the above bound implies that
\begin{equation*}
    \BP^r\left(\bm{1} + \beta^{-1} z \in \mathcal{V}^{(1-\beta^2)}(f)\right) \leq \ee^{-(1-\beta)^2 d/ 8r^2}\,.
\end{equation*}
Since $z$ is uniform on the $r$-ball, then $\beta^{-1} z$ becomes uniform on the $\beta^{-1} r$-ball. Thus the above bound reduces to
\begin{equation*}
    \BP^r\left(\bm{1} + z \in \mathcal{V}^{(1-\beta^2)}(f)\right) \leq \ee^{-\left(\frac{1-\beta}{\beta}\right)^2 d/8r^2}\,.
\end{equation*}
By changing the measure from $\BP^r$ to $\BP^r_\kappa$ the inequality in~\eqref{eq:CRU} follows.
\qed


\subsection{Concentration and Isoperimetric Inequalities}
Isoperimetric inequalities provide lower bounds for the surface measure of Borel subsets. Specifically, suppose $\mu$ is a given probability measure on $\BR^d$, and let $A \subset \BR^d$ be a Borel subset, whose $\delta$-extension is denoted by $A^\delta = A + \delta \mathbb{B}$. Then, the \emph{Minkowski content} (denoted by $\mu^+$) of the subset $A$ relative to the  measure $\mu$ is defined by 
\begin{equation}
\label{eq:Minkowski_content}
    \mu^+(A)\coloneq  \liminf_{\delta \to 0}\,\frac{\mu(A^\delta)-\mu(A)}{\delta}\,.
\end{equation}
Given this definition, we can think of $\mu^+(A)$ as the area measure of the boundary of $A$.

Among subsets with measures in a certain range, the \emph{isoperimetric function} $\mathcal{I}_\mu:[0,1/2) \to \BR_+$ returns the Minkowski content of the subset with the smallest boundary area. Formally, it is defined by
\begin{equation}
\label{eq:isop_function}
    \mathcal{I}_{\mu}(a) \coloneq  \inf_{1/2 < \mu(A)\leq 1-a} \mu
    ^+(A)\,.
\end{equation}
In many environments, where the probability measures satisfy some mild regularity conditions, there exist universal lower bounds for the isoperimetric function. One particular case that is of interest to us is the following lemma.
\begin{lemma}[\cite{Barthe2009}]
Let $u$ be the uniform measure on the probability surface $\Delta_d$. That is ${u}(A) = \mathsf{Vol}(A)/\mathsf{Vol}(\Delta_d)$, for every $A\subseteq \Delta_d$. Then, there exists a universal constant $c>0$ such that for $a\in [0,1/2)$:
\begin{equation}
\label{eq:Barthe_lower_bound}
    \mathcal{I}_{u}(a) \geq c \, a\, d\,. 
\end{equation}
\end{lemma}
In the following, we use $u(\cdot)$ to refer to the uniform measure on $\Delta_d$. As a corollary of the previous lemma we show that if a subset covers at least half of the measure on $\Delta_d$, then the uniform measure of its $\delta$-extension is very close to $1$.
\begin{corollary}
\label{cor:Minkowski_lower_bound}
Assume $u(A) \geq 1/2$, then
\begin{equation}
\label{eq:measure_lower_bound}
    u(A^\delta) \geq 1-\frac{1}{2}\,\ee^{-c \delta d}\,.
\end{equation}
\end{corollary}
\begin{proof}
    Because of the definition of the Minkowski content in~\eqref{eq:Minkowski_content} and the isoperimetric function in~\eqref{eq:isop_function}---both based on the limit inferior---one obtains 
    \begin{equation*}
        \begin{aligned}
        u(A^\delta) &\geq u(A)+\int_0^\delta u^+(A^t)\,\d t\\
        &\geq u(A) + \int_0^\delta \mathcal{I}_{u}(1-u(A^t))\, \d t\\
        &\geq u(A)+c d \int_0^\delta \big(1-u(A^t)\big)\,\d t\,,
        \end{aligned}
    \end{equation*}
    where the third inequality follows from~\eqref{eq:Barthe_lower_bound}. Define $z(0) \coloneq  u(A)$, and let $z:[0,\delta]\to \BR$ be the solution to the following integral equation:
    \begin{equation*}
        z(\delta) = z(0)+c d \int_0^\delta (1-z(t)) \, \d t\,.
    \end{equation*}
    Gr\"{o}nwall's inequality implies that $u(A^\delta) \geq z(\delta)$. One can simply verify that $z(\delta)=1-\left((1-z(0)\right)\ee^{-c \delta d }$, and this establishes the claim in~\eqref{eq:measure_lower_bound}.
\end{proof}
An important consequence of this result, that lies at the core of the proof of Theorem~\ref{thm:small_belief_set}, is that if two subsets in $\Delta_d$ have positive distance from each other, then the volume of at least one of them must be exponentially smaller than $\mathsf{Vol}(\Delta_d)$. Intuitively, this resembles the separation argument in Lemma~\ref{lem:high_dim_sep}, although its proof is not a direct consequence of the Brunn-Minkowski inequality, and follows from the more elaborate construct of the aforementioned isoperimetric lower bound. 
\begin{lemma}\label{lem:high_dim_sep_simplex}
Assume $A$ and $B$ are two Borel subsets of $\Delta_d$, where $A^\delta \cap B^\delta = \emptyset$. Then, there exists a universal constant $c>0$ such that:
\begin{equation}
\label{eq:min_vol_DeltaS}
    \min\{u(A), u(B)\} \leq \frac{1}{2}\,\ee^{-c \delta d}\,.
\end{equation}
\end{lemma}
\begin{proof}
Without any loss we shall assume that $u(A^\delta) \leq u(B^\delta)$. Since $A^\delta \cap B^\delta = \emptyset$, then $u(A^\delta) + u(B^\delta) \leq 1$. Hence the measure of the complement of $A^\delta$ is greater than or equal to $1/2$, i.e., $u(\Delta_d \setminus A^\delta) \geq 1/2$. We claim that $A \subseteq \Delta_d \, \setminus \, [\Delta_d \, \setminus \, A^\delta]^\delta$, which is equivalent to $\Delta d\, \setminus \, A \supseteq [\Delta_d \, \setminus \, A^\delta]^\delta$. To show the latter, pick any point $x\in [\Delta_d\, \setminus \, A^\delta]^\delta$. Recall that we defined the $\delta$-extension with strict inequality, hence, $\mathsf{dist}(x,\Delta_d \, \setminus \, A^\delta) <\delta$. Since $A^\delta$ is an open subset, then $\Delta_d \, \setminus \, A^\delta$ is compact, and thus there exists $y \in \Delta_d \, \setminus \, A^\delta$ such that
\begin{equation*}
    \norm{x-y}= \mathsf{dist}(x,\Delta_d \, \setminus \, A^\delta) < \delta \,.
\end{equation*}
On the other hand, $y \in \Delta_d \, \setminus \, A^\delta$ implies that $\mathsf{dist}(y,A) \geq \delta$. The previous two inequalities imply that $x\notin A$, and hence our claim is verified. Therefore, we have $u(A) \leq 1- u\big([\Delta_d \, \setminus \, A^\delta]^\delta\big)$.
Since $u\big(\Delta_d \, \setminus \, A^\delta\big) \geq 1/2$, Corollary~\ref{cor:Minkowski_lower_bound} implies that 
\begin{equation*}
    u\big([\Delta_d \, \setminus \, A^\delta]^\delta\big) \geq 1-\frac{1}{2}\, \ee^{-c\delta d}\,,
\end{equation*}
thereby verifying the inequality in~\eqref{eq:min_vol_DeltaS}.
\end{proof}


\subsection{Proof of the Results in Section~\ref{sec:beliefs}}
Proposition~\ref{prop:approx_common_prior} implies that if an allocation is $\ep$-Pareto dominated, then the extension of the subjective belief sets have an empty intersection. Hence, any arbitrary split of the agents' index set $I$ into two groups results in two subsets whose extensions also have an empty intersection. Thus, we can employ Lemma~\ref{lem:high_dim_sep_simplex} to conclude that at least one of these subsets should have a small volume.

\subsubsection*{Proof of Proposition~\ref{prop:approx_common_prior}} Assume towards a contradiction that $\bigcap_{i\in I} \mathcal{B}_i(f_i)^\delta \neq \emptyset$, while $f$ is $\ep$-Pareto dominated. Choose $\eta \in \bigcap_{i\in I} \mathcal{B}_i(f_i)^\delta$. Since $f$ is $\ep$-Pareto dominated, there is $g \in \mathcal{F}_{\bar\omega}$ such that $(1-\ep)g_i \succ_i f_i$ for all $i$. By definition of the subjective belief set, and continuity one obtains that $\mu_i \cdot \big[(1-\ep)g_i-f_i\big] > 0$ for all $\mu_i \in \mathcal{B}_i(f_i)$. Choose $\tilde \mu_i \in \argmin\{\norm{\eta- \mu_i}\colon \mu_i \in \mathcal{B}_i(f_i)\}$. Observe that $\norm{\eta-\tilde \mu_i} < \delta$, hence
\begin{equation*}
\begin{aligned}
    \left| \eta \cdot \big[(1-\ep)g_i-f_i\big] - \tilde \mu_i \cdot \big[(1-\ep)g_i-f_i\big] \right| \leq \norm{\eta-\tilde \mu_i} \norm{g_i(1-\ep)-f_i}
    < \delta \norm{(1-\ep)g_i-f_i}\,.
\end{aligned}
\end{equation*}
Therefore, it holds that
\begin{equation*}
    \eta \cdot  \big[(1-\ep)g_i-f_i\big] > \tilde \mu_i \cdot \big[(1-\ep)g_i-f_i\big] -\delta \norm{(1-\ep)g_i-f_i} >  -\delta \norm{(1-\ep)g_i-f_i}\,,
\end{equation*}
where the last inequality follows because $\tilde \mu_i \cdot \big[g_i(1-\ep)-f_i\big] > 0$. Summing over all agents $i$ leads to
\begin{equation*}
    \eta \cdot \sum_{i\in I} \big[(1-\ep)g_i-f_i\big] > -\delta \sum_{i\in I} \norm{(1-\ep)g_i-f_i} \geq  -\delta \bar\omega \rho =-\ep \bar\omega\,.
\end{equation*}
Since both $f$ and $g$ belong to $\mathcal{F}_{\bar\omega}$, the leftmost side above is equal to $-\ep \bar\omega$, thus leading to a contradiction.\qed

\subsubsection*{Proof of Theorem~\ref{thm:small_belief_set}} By Proposition~\ref{prop:approx_common_prior}, if $f$ is $\ep$-Pareto dominated, then $\bigcap_{i \in I} \mathcal{B}_i(f_i)^\delta = \emptyset$, where $\delta = \ep/\rho$. For an arbitrary $J \subset I$, one can readily verify that $\mathcal{B}_J^\delta \subseteq \bigcap_{j \in J}\mathcal{B}_j(f_j)^\delta$. Therefore, $\mathcal{B}_J^\delta \cap \mathcal{B}_{J^c}^\delta=\emptyset$. Under the $\ell_2$-norm, equation~\eqref{eq:rho} implies that $\rho = 2\sqrt{d}$, and thus $\delta = \ep/2\sqrt{d}$. Applying Lemma~\ref{lem:high_dim_sep_simplex} together with the previous expression for $\delta$ lead to the conclusion in~\eqref{eq:small_belief_set}.\qed

\section{Remaining Proofs}\label{sec:remainingproofs}
\subsection{Gaussian Measure on \texorpdfstring{$\mathbb{B}(r)$}.}
\label{app:constrained_gaussian}
In Section~\ref{sec:notationconventions} we remarked that the constant $\kappa$ for a conditional Gaussian measure is independent of $d$. Here we include a proof of this fact.

Let $\gamma_d \sim \mathcal{N}(0, \bm{I}_d)$ be the standard $d$-dimensional Gaussian measure, and $\widetilde{\BP}^r$ be its restriction to the Euclidean $r$-ball. Then the Radon-Nikodym derivative of $\widetilde{\BP}^r$ with respect to the uniform measure $\BP^r$ on the $r$-ball is
\begin{equation*}
\begin{gathered}
    \frac{\d \widetilde{\BP}^r }{\d \BP^r}(z) = \frac{\d \gamma_d(z)}{\d z} \,\,\frac{\int_{\mathbb{B}(r)} \d x}{\int_{\mathbb{B}(r)} \d \gamma_d(x)}
    = \ee^{-\norm{z}^2/ 2} \,\,\frac{\int_{\mathbb{B}(r)} \d x}{\int_{\mathbb{B}(r)}\ee^{-\norm{x}^2/2}\,\d x}
    \leq \frac{\int_{\mathbb{B}(r)} \d x}{\int_{\mathbb{B}(r)}\ee^{-\norm{x}^2/2}\,\d x}\,.
\end{gathered}
\end{equation*}
Let us denote the radial component of $x$ by $\rho$, i.e., $\rho = \norm{x}$. Then, $0 \leq \rho \leq r$ and one has
\begin{equation*}
    \frac{\int_{\mathbb{B}(r)} \d x}{\int_{\mathbb{B}(r)}\ee^{-\norm{x}^2/2}\,\d x} = \frac{\int_0^r \rho^{d-1}\, \d \rho}{\int_0^r \ee^{-\rho^2/2}\rho^{d-1}\, \d \rho} \leq \ee^{r^2/2}\,,
\end{equation*}
where in the last inequality, $\ee^{-\rho^2/2}$ (in the integrand of the denominator) is lower bounded by $\ee^{-r^2/2}$. Hence, the Radon-Nikodym derivative of the restricted standard Gaussian on the $r$-ball relative to the uniform measure is uniformly bounded above by $\ee^{r^2/2}$.

\subsection{Proof of Proposition~\ref{prop:theta_upper_bound}}
Since $\Pi$ is a $(d-1)$-dimensional surface with \emph{constant} width, then~\cite{schramm1988volume} implies that
\begin{equation*}
    \mathsf{Vol}(\Pi) \geq \left(\sqrt{3 + \frac{2}{d}} - 1\right)^{d-1} \mathsf{Vol}\left(\mathbb{B}^{d-1}(\bm{0},\theta/2)\right) \geq \left(\sqrt{3}-1\right)^{d-1} \left(\frac{\theta}{2}\right)^{d-1}\mathsf{Vol}\left(\mathbb{B}^{d-1}\right)\,. 
\end{equation*}
By Theorem~\ref{thm:small_belief_set} one has $\mathsf{Vol}(\Pi) \leq \frac{1}{2}\,\ee^{-c \ep/\sqrt{d}}\, \mathsf{Vol}(\Delta_{d})$, therefore the above inequality leads to
\begin{equation*}
    \left(\frac{\theta}{2}\right)^{d-1} \leq  \frac{\ee^{-c\ep \sqrt{d}}}{2\left(\sqrt{3}-1\right)^{d-1}}  \, \frac{\mathsf{Vol}(\Delta_{d})}{\mathsf{Vol}\left(\mathbb{B}^{d-1}\right)} \,. 
\end{equation*}
The volume of the $(d-1)$-dimensional unit $\ell_2$ ball and the $(d-1)$-dimensional probability simplex $\Delta_d$ are respectively equal to $\mathsf{Vol}(\mathbb{B}_2^{d-1}) = \pi^{(d-1)/2}/\Gamma\left(\frac{d+1}{2}\right)$ and $\mathsf{Vol}(\Delta_{d}) = \sqrt{d}/\Gamma(d)$. Therefore, 
\begin{equation*}
    \left(\frac{\theta}{2}\right)^{d-1} \leq \frac{\sqrt{d}\,\ee^{-c\ep \sqrt{d}}}{2\left(\sqrt{\pi}\big(\sqrt{3}-1\big)\right)^{d-1}} \, \frac{\Gamma\left(\frac{d+1}{2}\right)}{\Gamma(d)}\,.
\end{equation*}
Since the Gamma function is log-convex, then $\Gamma\left(\frac{d+1}{2}\right) \leq \sqrt{\Gamma(d)} = \sqrt{(d-1)!}$. Hence the above inequality simplifies to
\begin{equation*}
    \left(\frac{\theta}{2}\right)^{d-1} \leq \frac{d\,\ee^{-c\ep \sqrt{d}}}{2\left(\sqrt{\pi}\big(\sqrt{3}-1\big)\right)^{d-1}}\,\frac{1}{\sqrt{d!}}\,.
\end{equation*}
Let us denote $\sqrt{\pi}\big(\sqrt{3}-1\big)$ by $\alpha$. Since the width $\theta$ is smaller than $2$, then $(\theta/2)^{d} \leq (\theta/2)^{d-1}$ and thus
\begin{equation*}
    \theta \leq \frac{2}{\alpha} \left(\frac{\alpha d}{2}\right)^{1/d} \ee^{-c\ep /\sqrt{d}} (d!)^{-1/2d}\,.
\end{equation*}
We can readily verify that $\frac{2}{\alpha} \left(\frac{\alpha d}{2}\right)^{1/d} \leq 4$ for all integers $d$, thereby proving the inequality~\eqref{eq:theta_upper_bound}.\qed

\clearpage
\setcitestyle{numbers}	 
\bibliographystyle{normalstyle.bst}
\bibliography{ref}
\end{document}